\documentclass[runningheads,a4paper]{llncs}

\usepackage{amssymb}
\usepackage{graphicx}
\usepackage{epstopdf}
\usepackage{bussproofs}
\usepackage{lscape}
\usepackage{tikz-cd}
\usepackage{amsmath}
\usepackage{soul}
\usepackage{cmll}

\usepackage{url}
\urldef{\mailsa}\path|matthijs.vakar@ox.ac.uk|

\newcommand{\ra}[1]{\stackrel{#1}{\longrightarrow}}
\newcommand{\txt}[1]{\quad\textnormal{#1}\quad}
\newcommand{\sem}[2][M\!,g]{ [\![ #2 ]\!]^{}}
\newcommand{\lbi}[3]{\mathrm{let}\;#1\;\mathrm{be}\;#2\;\mathrm{in}\;#3}

\begin{document}

\mainmatter  % start of an individual contribution

% first the title is needed
\title{A Categorical Semantics\\
for Linear Logical Frameworks}

% a short form should be given in case it is too long for the running head
\titlerunning{A Categorical Semantics for Linear Logical Frameworks}

% the name(s) of the author(s) follow(s) next
%
% NB: Chinese authors should write their first names(s) in front of
% their surnames. This ensures that the names appear correctly in
% the running heads and the author index.
%
\author{Matthijs V\'ak\'ar}
\authorrunning{Matthijs V\'ak\'ar}
% (feature abused for this document to repeat the title also on left hand pages)

% the affiliations are given next; don't give your e-mail address
% unless you accept that it will be published
\institute{Department of Computer Science,\\
University of Oxford, Oxford, United Kingdom\\
%\\
%\url{http://www.cs.ox.ac.uk/}
}

%
% NB: a more complex sample for affiliations and the mapping to the
% corresponding authors can be found in the file "llncs.dem"
% (search for the string "\mainmatter" where a contribution starts).
% "llncs.dem" accompanies the document class "llncs.cls".
%

\toctitle{Syntax and Semantics of Linear Dependent Types}
\tocauthor{Matthijs V\'ak\'ar}
\maketitle

\vspace{-10pt}
\begin{abstract}A type theory is presented that combines (intuitionistic) linear types with type dependency, thus properly generalising both intuitionistic dependent type theory and full linear logic. A syntax and complete categorical semantics are developed, the latter in terms of (strict) indexed symmetric monoidal categories with comprehension. Various optional type formers are treated in a modular way. In particular, we see that the historically much-debated multiplicative quantifiers and identity types arise naturally from categorical considerations. These new multiplicative connectives are further characterised by several identities relating them to the usual connectives from dependent type theory and linear logic. Finally, one important class of models, given by families with values in some symmetric monoidal category, is investigated in detail.
%\keywords{linear dependent types, dependent types, linear types, linear logic, categorical semantics, Curry-Howard correspondence}
\end{abstract}
\vspace{-20pt}
\section{Introduction}\vspace{-7pt}
Starting from Church's simply typed $\lambda$-calculus (or intuitionistic propositional type theory), two extensions branch off in perpendicular directions: \vspace{-5pt}
\begin{itemize}
\item[$\bullet$] following the Curry-Howard propositions-as-types interpretation \emph{dependent type theory} (DTT) \cite{martin1998intuitionistic} extends the simply typed $\lambda$-calculus from a proof-calculus of intuitionistic propositional logic to one for predicate logic;
\item[$\bullet$] \emph{linear logic} \cite{girard1987linear} gives a more detailed resource-sensitive analysis, exposing precisely how many times each assumption is used in proofs.
\end{itemize}
\vspace{-5pt}

A combined \emph{linear dependent type theory} is one of the interesting directions to explore to gain a more fine-grained understanding of \emph{homotopy type theory} \cite{hottbook} from a computer science point of view, explaining its flow of information. Indeed, many of the usual settings for computational semantics are naturally linear in character, either because they arise as $!$-co-Kleisli categories (coherence space and game semantics) or for more fundamental reasons (quantum computation).

Combining dependent types and linear types is a non-trivial task, however, and despite some work by various authors that we shall discuss, the precise relationship between the two systems remains poorly understood. The discrepancy between linear and dependent types is the following.\vspace{-5pt}
\begin{itemize}
\item[$\bullet$] The lack of structural rules in \emph{linear type theory} forces us to refer to each variable precisely once---for a sequent $x:A\vdash t:B$, $x$ occurs uniquely in $t$.
\item[$\bullet$] In \emph{dependent type theory}, types can have free variables---$x:A\vdash B\;\mathrm{type}$, where $x$ is free in $B$. Crucially, if $x:A\vdash t:B$, $x$ may also be free in $t$.
\end{itemize}
\vspace{-7pt}\nopagebreak
What does it mean for $x$ to occur uniquely in $t$ in a dependent setting? Do we count its occurrence in $B$? The usual way out, which we shall also follow, is to restrict type dependency to intuitionistic terms. Although this seems very limiting---for instance, we do not obtain an analogue of the Girard translation that embeds DTT in the resulting system---it is not clear that there is a reasonable alternative. Moreover, as even this limited scenario has not been studied extensively, we hope that a semantic analysis, which was so far missing entirely, may shed new light on the old mystery of linear type dependency.

Historically, Girard's early work in linear logic already takes steps toward extending a linear analysis to predicate logic. Although it talks about first-order quantifiers, the analysis appears to have stayed rather superficial, omitting the identity predicates, which, in a way, are what make first-order logic tick. Closely related is that an account of internal quantification, or a linear variant of Martin-L\"of's type theory, was missing, let alone a Curry-Howard correspondence.

Later, linear types and dependent types were first combined in a Linear Logical Framework \cite{cervesato1996linear}, where a syntax was presented that extends a Logical Framework with linear types (that depend on terms of intuitionistic types). This has given rise to a line of work in the computer science community \cite{dal2011linear,petit2012linear,gaboardi2013linear}. All the work seems to be syntactic in nature, however, and seems to be mostly restricted to the asynchronous fragment in which we only have $\multimap$-, $\Pi$-, $\top$-, and $\&$-types. An exception is the Concurrent Logical Framework \cite{watkins2003concurrent}, which treats synchronous connectives resembling our $I$-, $\otimes$-, \mbox{$\Sigma$-,} and $!$-types. An account of additive disjunctions and identity types is missing entirely.

On the other hand, similar ideas, this time at the level of categorical semantics and specific models (from homotopy theory, algebra, and physics), have emerged in the mathematical community  \cite{may2006parametrized,shulman2013enriched,ponto2012duality,schreiber2014quantization}. In these models, as with Girard, a notion of comprehension was missing and, with that, a notion of identity type. Although, in the past year, some suggestions have been made on the nLab and nForum about possible connections between the syntactic and semantic work, no account of the correspondence was published, as far as the author is aware.

The point of this paper\footnote{This paper is based on the technical report \cite{vakar2014syntax}, where proofs and more discussion can be found. Independently, Krishnaswami et al. \cite{krishnaswamiintegrating} developed a roughly equivalent syntax and gave an operational rather than a denotational semantics. There, type dependency is added to Benton's LNL calculus, rather than to DILL.} is to close this gap between syntax and semantics and to pave the way for a proper semantic analysis of linear type dependency, treating a range of type formers including the crucial $\mathrm{Id}$-types\footnote{To be precise: extensional $\mathrm{Id}$-types. Intensional $\mathrm{Id}$-types remain a topic of investigation, due to the subtlety of dependent elimination rules in a linear setting.}. First, in Section \ref{sec:syn}, we present a syntax, intuitionistic linear dependent type theory (ILDTT), a natural blend of the dual intuitionistic linear logic (DILL) \cite{barber1996dual} and dependent type theory (DTT) \cite{hofmann1997syntax} which generalises both. Second, in Section \ref{sec:sem}, we present a complete categorical semantics, an obvious combination of linear/non-linear adjunctions \cite{barber1996dual} and comprehension categories  \cite{jacobs1993comprehension}. Finally, in Section \ref{sec:dismod}, an important class of models is studied: families with values in a symmetric monoidal category.
\clearpage
\section{Syntax}
\vspace{-7pt}
\label{sec:syn}
We assume the reader has some familiarity with the formal syntax of dependent type theory and linear type theory. In particular, we will not go into syntactic details like $\alpha$-conversion, name binding, capture-free substitution of $a$ for $x$ in $t$ (write $t[a/x]$), and pre-syntax. Details on all of these topics can be found in \cite{hofmann1997syntax}.

We next present the formal syntax of ILDTT. We start with a presentation of the judgements that represent the propositions of the language and then discuss its rules of inference: first its structural core, then the logical rules for a series of optional type formers. We conclude this section with a few basic results about the syntax.

\subsubsection*{Judgements}
We adopt a notation $\Delta;\Xi$ for contexts, where $\Delta$ is `an intuitionistic region' and $\Xi$ is `a linear region', as in DILL \cite{barber1996dual}. The idea is that we have an empty context and can extend an existing context $\Delta;\Xi$ with both intuitionistic and linear types that are allowed to depend on $\Delta$.

Our language expresses judgements of the following six forms.
\vspace{-10pt}
\begin{figure}
\centering
\fbox{\parbox{\textwidth}{\scriptsize
\begin{tabular}{ll}
\textbf{ILDTT judgement} & \textbf{Intended meaning}\vspace{2pt}\\
$\vdash \Delta;\Xi \;\mathrm{ctxt}$ & $\Delta;\Xi$ is a valid context\\
$\Delta;\cdot \vdash A\;\mathrm{type}$ &  $A$ is a type in (intuitionistic) context $\Delta$\\
$\Delta;\Xi\vdash a:A$ & $a$ is a term of type $A$ in context $\Delta;\Xi$\\
$\vdash \Delta;\Xi \equiv \Delta';\Xi'\;\mathrm{ctxt}$\hspace{20pt} & $\Delta;\Xi$ and $\Delta';\Xi'$ are judgementally equal contexts\\
$\Delta;\cdot\vdash A\equiv A'\;\mathrm{type}$ & $A$ and $A'$ are judgementally equal types in (intuitionistic) context $\Delta$\\
$\Delta;\Xi\vdash a\equiv a':A$ & $a$ and $a'$ are judgementally equal terms of type $A$ in context $\Delta;\Xi$
\end{tabular}}}
\normalsize
\caption{Judgements of ILDTT.}
\end{figure}
\vspace{-20pt}

\subsubsection*{Structural Rules}
We will use the following structural rules, which are essentially the structural rules of dependent type theory where some rules appear in both an intuitionistic and a linear form. We present the rules by group, with their names, from left to right and top to bottom.
\vspace{-10pt}
\begin{figure}
\centering
\fbox{\parbox{\textwidth}{\scriptsize
Rules for context formation (C-Emp, Int-C-Ext, Int-C-Ext-Eq, Lin-C-Ext, Lin-C-Ext-Eq):\\
\\
\begin{tabular}{lr}
\AxiomC{}
\RightLabel{}
\UnaryInfC{$\vdash \cdot;\cdot \;\mathrm{ctxt}$}
\DisplayProof
& \\
& \\
\AxiomC{$\vdash\Delta;\Xi\; \mathrm{ctxt}$}
\AxiomC{$\Delta;\cdot \vdash A\;\mathrm{type}$}
\RightLabel{}
\BinaryInfC{$\vdash \Delta,x:A;\Xi \;\mathrm{ctxt}$}
\DisplayProof

&
\hspace{56pt}
\AxiomC{$\vdash \Delta;\Xi\equiv\Delta';\Xi'\;\mathrm{ctxt}$}
\AxiomC{$\Delta;\cdot \vdash A\equiv B\;\mathrm{type}$}
\RightLabel{}
\BinaryInfC{$\vdash \Delta,x:A;\Xi\equiv\Delta',y:B;\Xi'\;\mathrm{ctxt}$}
\DisplayProof\\
& \\
\AxiomC{$\vdash \Delta;\Xi\;\mathrm{ctxt}$}
\AxiomC{$\Delta;\cdot \vdash A\;\mathrm{type}$}
\RightLabel{}
\BinaryInfC{$\vdash \Delta;\Xi,x:A\;\mathrm{ctxt}$}
\DisplayProof

&

\AxiomC{$\vdash \Delta;\Xi\equiv\Delta';\Xi'\;\mathrm{ctxt}$}
\AxiomC{$\Delta;\cdot\vdash A\equiv B\;\mathrm{type}$}
\RightLabel{}
\BinaryInfC{$\vdash \Delta;\Xi,x:A\equiv\Delta';\Xi',y:B\;\mathrm{ctxt}$}
\DisplayProof
\end{tabular}\\
\\
\\
Variable declaration/axiom rules (Int-Var, Lin-Var):\\
\\
\begin{tabular}{lr}
\AxiomC{$\vdash \Delta,x:A,\Delta';\cdot\;\mathrm{ctxt}$}
\RightLabel{}
\UnaryInfC{$\Delta,x:A,\Delta';\cdot\vdash x:A$}
\DisplayProof
&\hspace{188pt}
\AxiomC{$\vdash \Delta;x:A\;\mathrm{ctxt}$}
\RightLabel{}
\UnaryInfC{$\Delta;x:A\vdash x:A$}
\DisplayProof
\end{tabular}}}
\normalsize
\caption{Context formation and variable declaration rules.}
\end{figure}

\vspace{-30pt}
\begin{figure}
\centering
\fbox{\parbox{\textwidth}{\scriptsize
The standard rules expressing that judgemental equality is an equivalence relation (C-Eq-R, C-Eq-S, C-Eq-T, Ty-Eq-R, Ty-Eq-S, Ty-Eq-T, Tm-Eq-R, Tm-Eq-S, Tm-Eq-T):\\
\\
\begin{tabular}{lr}
\AxiomC{$\vdash \Delta;\Xi\;\mathrm{ctxt}$}
\RightLabel{}
\UnaryInfC{$\vdash \Delta;\Xi\equiv \Delta;\Xi\;\mathrm{ctxt}$}
\DisplayProof
&\hspace{65pt}
\AxiomC{$\vdash \Delta;\Xi\equiv \Delta';\Xi'\;\mathrm{ctxt}$}
\RightLabel{}
\UnaryInfC{$\vdash \Delta';\Xi'\equiv \Delta;\Xi\;\mathrm{ctxt}$}
\DisplayProof\\
&\\
\AxiomC{$\vdash \Delta;\Xi\equiv \Delta';\Xi'\;\mathrm{ctxt}$}
\AxiomC{$\vdash \Delta';\Xi'\equiv \Delta'';\Xi''\;\mathrm{ctxt}$}
\RightLabel{}
\BinaryInfC{$\vdash \Delta;\Xi\equiv \Delta'';\Xi''\;\mathrm{ctxt}$}
\DisplayProof &\\
&\\
\AxiomC{$\Delta;\cdot\vdash A\;\mathrm{type}$}
\RightLabel{}
\UnaryInfC{$\Delta;\cdot\vdash A\equiv A\;\mathrm{type}$}
\DisplayProof
&
\AxiomC{$\Delta;\cdot\vdash A\equiv A'\;\mathrm{type}$}
\RightLabel{}
\UnaryInfC{$\Delta;\cdot\vdash A'\equiv A\;\mathrm{type}$}
\DisplayProof\\
&\\
\AxiomC{$\Delta;\cdot\vdash A\equiv A'\;\mathrm{type}$}
\AxiomC{$\Delta;\cdot\vdash A'\equiv A''\;\mathrm{type}$}
\RightLabel{}
\BinaryInfC{$\Delta;\cdot\vdash A\equiv A''\;\mathrm{type}$}
\DisplayProof
&\\
&\\
\AxiomC{$\Delta;\Xi\vdash a:A$}
\RightLabel{}
\UnaryInfC{$\Delta;\Xi\vdash a\equiv a: A$}
\DisplayProof
&
\AxiomC{$\Delta;\Xi\vdash a\equiv a':A$}
\RightLabel{}
\UnaryInfC{$\Delta;\Xi\vdash a'\equiv a: A$}
\DisplayProof
\\
&\\
\AxiomC{$\Delta;\Xi\vdash a\equiv a':A$}
\AxiomC{$\Delta;\Xi\vdash a'\equiv a'':A$}
\RightLabel{}
\BinaryInfC{$\Delta;\Xi\vdash a\equiv a'': A$}
\DisplayProof
\end{tabular}
\\
\\
\\
The standard rules relating typing and judgemental equality (Tm-Conv, Ty-Conv):\\
\\
\begin{tabular}{ll}
\AxiomC{$\Delta;\Xi\vdash a:A$}
\AxiomC{$\vdash \Delta;\Xi\equiv \Delta';\Xi'\;\mathrm{ctxt}$}
\AxiomC{$\Delta;\cdot \vdash A\equiv A'\;\mathrm{type}$}
\RightLabel{}
\TrinaryInfC{$\Delta';\Xi'\vdash a:A'$}
\DisplayProof &\\
&\\
\AxiomC{$\Delta;\cdot\vdash A\;\mathrm{type}$}
\AxiomC{$\vdash \Delta;\cdot\equiv \Delta';\cdot\;\mathrm{ctxt}$}
\RightLabel{}
\BinaryInfC{$\Delta';\cdot\vdash A\;\mathrm{type}$}
\DisplayProof
\end{tabular}
\normalsize}}
\caption{A few standard rules for judgemental equality.}
\end{figure}
\begin{figure}
\centering
\fbox{\parbox{\textwidth}{\scriptsize
Exchange, weakening, and substitution rules (Int-Weak, Int-Exch, Lin-Exch, Int-Ty-Subst, Int-Ty-Subst-Eq, Int-Tm-Subst, Int-Tm-Subst-Eq, Lin-Tm-Subst, Lin-Tm-Subst-Eq):\\
\\
\begin{tabular}{lr}
\AxiomC{$\Delta,\Delta';\Xi\vdash\mathcal{J}$}
\AxiomC{$\Delta;\cdot\vdash A\;\mathrm{type}$}
\RightLabel{}
\BinaryInfC{$\Delta,x:A,\Delta';\Xi\vdash \mathcal{J}$}
\DisplayProof
&\\
&\\
&\\
\AxiomC{$\Delta,x:A,x':A',\Delta';\Xi\vdash \mathcal{J}$}
\RightLabel{}
\UnaryInfC{$\Delta,x':A',x:A,\Delta';\Xi\vdash \mathcal{J}$}
\DisplayProof
&\vspace{4pt}
\AxiomC{$\Delta;\Xi,x:A,x':A',\Xi'\vdash \mathcal{J}$}
\RightLabel{}
\UnaryInfC{$\Delta;\Xi,x':A',x:A,\Xi'\vdash \mathcal{J}$}
\DisplayProof
\\
(if $x$ is not free in $A'$)&\\
&\\
\AxiomC{${\Delta},x:A,\Delta';\cdot \vdash B\;\mathrm{type}$}
\AxiomC{$\Delta;\cdot \vdash a:A$}
\RightLabel{}
\BinaryInfC{${\Delta},\Delta'[{a}/x];\cdot \vdash B[{a}/x]\;\mathrm{type}$}
\DisplayProof
&\hspace{13pt}
\AxiomC{${\Delta},x:A,\Delta';\cdot \vdash B\equiv B'\;\mathrm{type}$}
\AxiomC{$\Delta;\cdot \vdash a:A$}
\RightLabel{}
\BinaryInfC{${\Delta},\Delta'[{a}/x];\cdot \vdash B[{a}/x]\equiv B'[{a}/x]\;\mathrm{type}$}
\DisplayProof\\
&\\
&\\
\AxiomC{${\Delta},x:A,\Delta';\Xi \vdash b:B$}
\AxiomC{$\Delta;\cdot \vdash a:A$}
\RightLabel{}
\BinaryInfC{${\Delta},\Delta'[{a}/x];\Xi[{a}/x] \vdash b[{a}/x]:B[{a}/x]$}
\DisplayProof
&
\AxiomC{${\Delta},x:A,\Delta';\Xi \vdash b\equiv b':B$}
\AxiomC{$\Delta;\cdot \vdash a:A$}
\RightLabel{}
\BinaryInfC{${\Delta},\Delta'[{a}/x];\Xi[{a}/x] \vdash b[{a}/x]\equiv b'[{a}/x]:B[{a}/x]$}
\DisplayProof
\\
&\\
&\\
\AxiomC{$\Delta;\Xi,x:A\vdash b:B$}
\AxiomC{$\Delta;\Xi'\vdash a:A$}
\RightLabel{}
\BinaryInfC{$\Delta;\Xi,\Xi'\vdash b[a/x]:B$}
\DisplayProof
&
\AxiomC{$\Delta;\Xi,x:A\vdash b\equiv b':B$}
\AxiomC{$\Delta;\Xi'\vdash a:A$}
\RightLabel{}
\BinaryInfC{$\Delta;\Xi,\Xi'\vdash b[a/x]\equiv b'[a/x]:B$}
\DisplayProof\\
&\\
&\\
\end{tabular}\\
\vspace{-15pt}
\normalsize}}
\caption{Exchange, weakening, and substitution rules. Here, $\mathcal{J}$ represents a statement of the form $B\;\mathrm{type}$, $B\equiv B'$, $b:B$, or $b\equiv b':B$, such that all judgements are well-formed.}
\end{figure}
\vspace{-30pt}
\clearpage
\subsubsection*{Logical Rules} We describe some (optional) type and term formers, for which we give type formation (denoted -F), introduction (-I), elimination (-E), computation rules (-C), and (judgemental) uniqueness principles (-U). We also assume the obvious rules stating that the type formers and term formers respect judgemental equality. Moreover, $\Sigma_{!x:!A}$, $\Pi_{!x:!A}$, $\lambda_{!x:!A}$, and $\lambda_{x:A}$ are name binding operators, binding free occurrences of $x$ within their scope.

We demand -U-rules for the various type formers in this paper, as this allows us to give a natural categorical semantics. This includes $\mathrm{Id}$-types: we study extensional identity types. In practice, when building a computational implementation of a type theory like ours, one would probably drop some of these rules to make the system decidable, which would correspond to switching to weak equivalents of the categorical constructions presented here.\footnote{In that case, in DTT, one would usually demand some stronger `dependent' elimination rules, which would make propositional equivalents of the -U-rules provable, adding some extensionality to the system, while preserving its computational properties. Such rules are problematic in ILDTT, however, both from a syntactic and semantic point of view and further investigation is warranted here.}
\vspace{-15pt}
\begin{figure}
\centering
\fbox{
\parbox{\textwidth}{\scriptsize
\begin{tabular}{lr}\begin{tabular}{l}
\AxiomC{$\Delta,x:A;\cdot\vdash B\;\mathrm{type}$}
\RightLabel{}
\UnaryInfC{$\Delta;\cdot\vdash \Sigma_{!x:{!A}}B\;\mathrm{type}$}
\DisplayProof\\
\\
\AxiomC{$\Delta;\cdot \vdash a:A$}
\AxiomC{$\Delta;\Xi \vdash b:B[{a}/x]$}
\RightLabel{}
\BinaryInfC{$\Delta ; \Xi \vdash  !{a} \otimes b:\Sigma_{!x:{!A}}B $}
\DisplayProof
\end{tabular}
&
\AxiomC{$\Delta;\cdot \vdash C\;\mathrm{type}$}
\noLine
\UnaryInfC{$\Delta;\Xi \vdash t:\Sigma_{!x:{!A}}B$}
\noLine
\UnaryInfC{$\Delta,x:A;\Xi',y:B\vdash c:C$}
\RightLabel{}
\UnaryInfC{$\Delta;\Xi,\Xi' \vdash \mathrm{let}\;t\;\mathrm{be}\;  !{x} \otimes y \;\mathrm{in}\;c:C$}
\DisplayProof 
\\[-2pt]
&{\scriptsize(if $\vdash\Delta;\Xi'\;\mathrm{ctxt}$)}
\\
\\
\AxiomC{$ \Delta;\Xi\vdash \mathrm{let}\; !{a} \otimes b\;\mathrm{be}\;  !{x} \otimes y \;\mathrm{in}\;c:C$}
\RightLabel{}
\UnaryInfC{$\Delta;\Xi\vdash \mathrm{let}\; !{a} \otimes b\;\mathrm{be}\; ! {x} \otimes y \;\mathrm{in}\;c\equiv c[{a}/x,b/y]:C$}
\DisplayProof\hspace{-10pt}
&
\hspace{-3pt}
\AxiomC{$ \Delta;\Xi\vdash \mathrm{let}\;t\;\mathrm{be}\;  !{x} \otimes y \;\mathrm{in}\; !{x} \otimes y:\Sigma_{!x:{!A}}B$}
\RightLabel{}
\UnaryInfC{$\Delta;\Xi\vdash \mathrm{let}\;t\;\mathrm{be}\; ! {x} \otimes y \;\mathrm{in}\; !{x} \otimes y\equiv t:\Sigma_{!x:{!A}}B$}
\DisplayProof
\\
&\\
&\\
\AxiomC{$\Delta,x:A;\cdot \vdash B\;\mathrm{type}$}
\RightLabel{}
\UnaryInfC{$\Delta;\cdot\vdash\Pi_{!x:!{A}}B\;\mathrm{type}$}
\DisplayProof
&\\
&\\
\AxiomC{$\vdash \Delta;\Xi\;\mathrm{ctxt}$}
\AxiomC{$\Delta,x:A;\Xi\vdash b:B$}
\RightLabel{}
\BinaryInfC{$\Delta;\Xi\vdash \lambda_{!x:{!A}}b:\Pi_{!x:!{A}}B$}
\DisplayProof
&
\AxiomC{$\Delta;\cdot \vdash a:A$}
\AxiomC{$\Delta;\Xi\vdash f:\Pi_{!x:{!A}}B$}
\RightLabel{}
\BinaryInfC{$\Delta;\Xi\vdash f(!{a}):B[{a}/x]$}
\DisplayProof\\
&\\
&\\
\AxiomC{$\Delta;\Xi\vdash (\lambda_{!x:{!A}}b)(!{a}):B[{a}/x]$}
\RightLabel{}
\UnaryInfC{$\Delta;\Xi\vdash (\lambda_{!x:!{A}}b)(!{a})\equiv b[{a}/x]:B[{a}/x]$}
\DisplayProof
&
\AxiomC{$\Delta;\Xi\vdash \lambda_{!x:{!A}}f(!x):\Pi_{!x:{!A}}B$}
\RightLabel{}
\UnaryInfC{$\Delta;\Xi\vdash f\equiv \lambda_{!x:{!A}}f(!x):\Pi_{!x:{!A}}B$}
\DisplayProof
\\
&\\
&\\
\begin{tabular}{l}
\AxiomC{$\Delta;\cdot \vdash a:A$}
\AxiomC{$\Delta;\cdot \vdash a':A$}
\RightLabel{}
\BinaryInfC{$\Delta;\cdot \vdash \mathrm{Id}_{!A}(a,a')\;\mathrm{type}$}
\DisplayProof \\
\\
\AxiomC{$\Delta;\cdot \vdash a:A$}
\RightLabel{}
\UnaryInfC{$\Delta;\cdot \vdash \mathrm{refl}_{!a}:\mathrm{Id}_{!A}(a,a)$}
\DisplayProof
\end{tabular}\hspace{-64pt}
&\hspace{-62pt} \begin{tabular}{l}
\hspace{75pt}$\Delta,x:A,x':A;\cdot \vdash D\;\mathrm{type}$\\
\hspace{75pt}$\Delta,z:A;\Xi \vdash d:D[{z}/x,{z}/x']$\\
\hspace{75pt}$\Delta ;\cdot \vdash a:A$\\
\hspace{75pt}$\Delta;\cdot \vdash a':A$\\
\AxiomC{$\Delta ;\Xi' \vdash p:\mathrm{Id}_{!A}(a,a')$}
\RightLabel{}
\UnaryInfC{$\Delta;\Xi[a/z],\Xi' \vdash \mathrm{let}\; (a,a',p)\;\mathrm{be}\;(z,z,\mathrm{refl}_{!z})\;\mathrm{in}\; d:D[{a}/x,{a'}/x']$}
\DisplayProof
\end{tabular}
\\
&\\
&\\
\AxiomC{$\Delta;\Xi[{a}/z]\vdash \mathrm{let}\; (a,a,\mathrm{refl}_{!a})\;\mathrm{be}\;(z,z,\mathrm{refl}_{!z})\;\mathrm{in}\; d:D[{a}/x,{a}/x']$}
\RightLabel{}
\UnaryInfC{$\Delta;\Xi[{a}/z]\vdash \mathrm{let}\; (a,a,\mathrm{refl}_{!a})\;\mathrm{be}\;(z,z,\mathrm{refl}_{!z})\;\mathrm{in}\; d\equiv d[{a}/z] :D[{a}/x,{a}/x']$}
\DisplayProof \hspace{-85pt}
&\hspace{-43pt}
\end{tabular}
\vspace{8pt}\\
\begin{tabular}{lr}
\AxiomC{$\Delta,x:A,x':A;\Xi, z:\mathrm{Id}_{!A}(x,x') \vdash \mathrm{let}\;(x,x',z)\;\mathrm{be}\;(x,x,\mathrm{refl}_{!x})\;\mathrm{in}\;c[x/x',\mathrm{refl}_{!x}/z]:C$}
\RightLabel{}
\UnaryInfC{$\Delta,x:A,x':A;\Xi, z:\mathrm{Id}_{!A}(x,x') \vdash \mathrm{let}\;(x,x',z)\;\mathrm{be}\;(x,x,\mathrm{refl}_{!x})\;\mathrm{in}\;c[x/x',\mathrm{refl}_{!x}/z]\equiv c:C$}
\DisplayProof 
\end{tabular}
\normalsize
}}
\caption{Rules for linear equivalents of some of the usual type formers from DTT ($\Sigma$-F, -I, -E, -C, -U, $\Pi$-F, -I, -E, -C, -U, $\mathrm{Id}$-F, -I, -E, -C, -U).}
\vspace{-10pt}
\end{figure}
\nopagebreak
\begin{figure}
\centering
\fbox{\parbox{\textwidth}{\scriptsize 
\begin{tabular}{lr}
\AxiomC{}
\RightLabel{}
\UnaryInfC{$\Delta;\cdot\vdash I\;\mathrm{type}$}
\DisplayProof
&\\
\AxiomC{}
\RightLabel{}
\UnaryInfC{$\Delta;\cdot\vdash *:I$}
\DisplayProof
&
\AxiomC{$\Delta;\Xi'\vdash t:I$}
\AxiomC{$\Delta;\Xi\vdash a:A$}
\RightLabel{}
\BinaryInfC{$\Delta;\Xi,\Xi'\vdash \mathrm{let}\;t\;\mathrm{be}\;*\;\mathrm{in}\;a:A$}
\DisplayProof 
 \\
&\\
\AxiomC{$\Delta;\Xi\vdash  \mathrm{let}\;*\;\mathrm{be}\;*\;\mathrm{in}\;a:A$}
\RightLabel{}
\UnaryInfC{$\Delta;\Xi\vdash \mathrm{let}\;*\;\mathrm{be}\;*\;\mathrm{in}\;a\equiv a :A$}
\DisplayProof
&
\AxiomC{$\Delta;\Xi\vdash \mathrm{let}\;t\;\mathrm{be}\;*\;\mathrm{in}\;*:I$}
\RightLabel{}
\UnaryInfC{$\Delta;\Xi\vdash \mathrm{let}\;t\;\mathrm{be}\;*\;\mathrm{in}\;*\equiv t :I$}
\DisplayProof 
\\
&\\
\AxiomC{$\Delta;\cdot \vdash A\;\mathrm{type}$}
\AxiomC{$\Delta;\cdot \vdash B\;\mathrm{type}$}
\RightLabel{}
\BinaryInfC{$\Delta;\cdot \vdash A\otimes B\;\mathrm{type}$}
\DisplayProof 
&\\
&\\
\AxiomC{$\Delta;\Xi\vdash a:A$}
\AxiomC{$\Delta;\Xi'\vdash b:B$}
\RightLabel{}
\BinaryInfC{$\Delta;\Xi,\Xi'\vdash a\otimes b:A\otimes B$}
\DisplayProof
&
\AxiomC{$\Delta;\Xi\vdash t:A\otimes B$}
\AxiomC{$\Delta;\Xi',x:A,y:B\vdash c:C$}
\RightLabel{}
\BinaryInfC{$\Delta;\Xi,\Xi'\vdash \mathrm{let}\; t\;\mathrm{be}\;x\otimes y\;\mathrm{in} \; c:C$}
\DisplayProof
\\
&\\
\AxiomC{$\Delta;\Xi \vdash \mathrm{let}\; a\otimes b \;\mathrm{be}\;x\otimes y\;\mathrm{in} \; c :C$}
\RightLabel{}
\UnaryInfC{$\Delta;\Xi \vdash \mathrm{let}\; a\otimes b \;\mathrm{be}\;x\otimes y\;\mathrm{in} \; c \equiv c[a/x,b/y]:C$}
\DisplayProof \hspace{-11pt}
&
\AxiomC{$\Delta;\Xi \vdash \mathrm{let}\; t \;\mathrm{be}\;x\otimes y\;\mathrm{in} \; x\otimes y :A\otimes B$}
\RightLabel{}
\UnaryInfC{$\Delta;\Xi \vdash \mathrm{let}\; t \;\mathrm{be}\;x\otimes y\;\mathrm{in} \; x\otimes y\equiv t:A\otimes B$}
\DisplayProof
\\
&\\
&\\
\AxiomC{$\Delta;\cdot \vdash A\;\mathrm{type}$}
\AxiomC{$\Delta;\cdot \vdash B\;\mathrm{type}$}
\RightLabel{}
\BinaryInfC{$\Delta;\cdot \vdash A\multimap B\;\mathrm{type}$}
\DisplayProof
&\\
&\\
\AxiomC{$\Delta;\Xi,x:A\vdash b:B$}
\RightLabel{}
\UnaryInfC{$\Delta;\Xi\vdash \lambda_{x:A}b:A\multimap B$}
\DisplayProof
&
\AxiomC{$\Delta;\Xi\vdash f:A\multimap B$}
\AxiomC{$\Delta;\Xi'\vdash a:A$}
\RightLabel{}
\BinaryInfC{$\Delta;\Xi,\Xi'\vdash f(a):B$}
\DisplayProof
\\
&\\
\AxiomC{$\Delta;\Xi \vdash (\lambda_{x:A}b)(a):B$}
\RightLabel{}
\UnaryInfC{$\Delta;\Xi\vdash (\lambda_{x:A}b)(a)\equiv b[a/x]:B$}
\DisplayProof
&
\AxiomC{$\Delta;\Xi \vdash \lambda_{x:A}fx:A\multimap B$}
\RightLabel{}
\UnaryInfC{$\Delta;\Xi\vdash \lambda_{x:A}fx\equiv f:A\multimap B$}
\DisplayProof
\end{tabular}
\\
\\
\\
\begin{tabular*}{\textwidth}{lcr}
\AxiomC{}
\RightLabel{}
\UnaryInfC{$\Delta;\cdot\vdash \top\;\mathrm{type}$}
\DisplayProof
\hspace{76pt}
&
\AxiomC{$\vdash \Delta;\Xi\;\mathrm{ctxt}$}
\RightLabel{}
\UnaryInfC{$\Delta;\Xi\vdash \langle\rangle:\top$}
\DisplayProof
\hspace{76pt}
&
\AxiomC{$\Delta;\Xi\vdash t:\top$}
\RightLabel{}
\UnaryInfC{$\Delta;\Xi\vdash t\equiv\langle\rangle:\top$}
\DisplayProof
\end{tabular*}
\\
\\
\\
\begin{tabular}{lr}
\AxiomC{$\Delta;\cdot \vdash A\;\mathrm{type}$}
\AxiomC{$\Delta;\cdot \vdash B\;\mathrm{type}$}
\RightLabel{}
\BinaryInfC{$\Delta;\cdot \vdash A\& B\;\mathrm{type}$}
\DisplayProof
\hspace{96pt}
&
\AxiomC{$\Delta;\Xi\vdash a:A$}
\AxiomC{$\Delta;\Xi\vdash b:B$}
\RightLabel{}
\BinaryInfC{$\Delta;\Xi\vdash \langle a, b\rangle:A\& B$}
\DisplayProof
\\
&\\
\AxiomC{$\Delta;\Xi\vdash t:A\& B$}
\RightLabel{}
\UnaryInfC{$\Delta;\Xi\vdash \mathrm{fst}(t):A$}
\DisplayProof
&
\AxiomC{$\Delta;\Xi\vdash t:A\& B$}
\RightLabel{}
\UnaryInfC{$\Delta;\Xi\vdash \mathrm{snd}(t):B$}
\DisplayProof
\\
&\\
\AxiomC{$\Delta;\Xi \vdash \mathrm{fst}(\langle a,b\rangle ):A$}
\RightLabel{}
\UnaryInfC{$\Delta;\Xi \vdash \mathrm{fst}(\langle a,b\rangle )\equiv a :A$}
\DisplayProof 
&
\AxiomC{$\Delta;\Xi \vdash \mathrm{snd}(\langle a,b\rangle):B$}
\RightLabel{}
\UnaryInfC{$\Delta;\Xi \vdash \mathrm{snd}(\langle a,b\rangle)\equiv b:B $}
\DisplayProof \\
&\\
&\\
\AxiomC{$\Delta;\Xi\vdash \langle\mathrm{fst}(t),\mathrm{snd}(t) \rangle:A\& B$}
\RightLabel{}
\UnaryInfC{$\Delta;\Xi\vdash \langle\mathrm{fst}(t),\mathrm{snd}(t) \rangle\equiv t:A\& B$}
\DisplayProof
&\\
&\\
\end{tabular}
\begin{tabular}{lcr}
\AxiomC{}
\RightLabel{}
\UnaryInfC{$\Delta ; \cdot\vdash 0\;\mathrm{type}$}
\DisplayProof
\hspace{54pt}
&
\AxiomC{$\Delta;\Xi\vdash t:0$}
\RightLabel{}
\UnaryInfC{$\Delta;\Xi,\Xi'\vdash \mathrm{false}(t) :B$}
\DisplayProof
\hspace{53pt}
&
\AxiomC{$\Delta;\Xi\vdash t:0$}
\RightLabel{}
\UnaryInfC{$\Delta;\Xi\vdash \mathrm{false}(t)\equiv t :0$}
\DisplayProof
\end{tabular}
\\
\\
\\
\begin{tabular}{lr}
\AxiomC{$\Delta;\cdot \vdash A\;\mathrm{type}$}
\AxiomC{$\Delta;\cdot \vdash B\;\mathrm{type}$}
\RightLabel{}
\BinaryInfC{$\Delta;\cdot \vdash A\oplus B\;\mathrm{type}$}
\DisplayProof
& \\
& \\
\AxiomC{$\Delta ;\Xi\vdash a: A$}
\RightLabel{}
\UnaryInfC{$\Delta;\Xi\vdash \mathrm{inl}(a): A\oplus B$}
\DisplayProof
&
\hspace{7pt}
\AxiomC{$\Delta ;\Xi\vdash b: B$}
\RightLabel{}
\UnaryInfC{$\Delta;\Xi\vdash \mathrm{inr}(b): A\oplus B$}
\DisplayProof
\\
&\\
&\\
\AxiomC{$\Delta ;\Xi,x:A\vdash c: C$}
\AxiomC{$\Delta ;\Xi,y:B\vdash d: C$}
\AxiomC{$\Delta ;\Xi'\vdash t:A\oplus B$}
\RightLabel{}
\TrinaryInfC{$\Delta;\Xi,\Xi'\vdash \mathrm{case}\; t\;\mathrm{of}\;\mathrm{inl}(x)\rightarrow c\;||\;\mathrm{inr}(y)\rightarrow d :C$}
\DisplayProof
&\\
&\\
\AxiomC{$\Delta;\Xi,\Xi'\vdash \mathrm{case}\; \mathrm{inl}(a)\;\mathrm{of}\;\mathrm{inl}(x)\rightarrow c\;||\;\mathrm{inr}(y)\rightarrow d :C$}
\RightLabel{}
\UnaryInfC{$\Delta;\Xi,\Xi'\vdash \mathrm{case}\; \mathrm{inl}(a)\;\mathrm{of}\;\mathrm{inl}(x)\rightarrow c\;||\;\mathrm{inr}(y)\rightarrow d \equiv c[a/x]:C$}
\DisplayProof
&\\
&\\
\AxiomC{$\Delta;\Xi,\Xi'\vdash \mathrm{case}\; \mathrm{inr}(b)\;\mathrm{of}\;\mathrm{inl}(x)\rightarrow c\;||\;\mathrm{inr}(y)\rightarrow d :C$}
\RightLabel{}
\UnaryInfC{$\Delta;\Xi,\Xi'\vdash \mathrm{case}\; \mathrm{inr}(b)\;\mathrm{of}\;\mathrm{inl}(x)\rightarrow c\;||\;\mathrm{inr}(y)\rightarrow d \equiv d[b/y]:C$}
\DisplayProof &\\
&\\
\AxiomC{$\Delta;\Xi\vdash \mathrm{case}\; t\;\mathrm{of}\;\mathrm{inl}(x)\rightarrow \mathrm{inl}(x)\;||\;\mathrm{inr}(y)\rightarrow \mathrm{inr}(y):A\oplus B $}
\RightLabel{}
\UnaryInfC{$\Delta;\Xi\vdash \mathrm{case}\; t\;\mathrm{of}\;\mathrm{inl}(x)\rightarrow \mathrm{inl}(x)\;||\;\mathrm{inr}(y)\rightarrow \mathrm{inr}(y)\equiv t:A\oplus B$}
\DisplayProof
\end{tabular}}}
\end{figure}
\clearpage
\begin{figure}
\centering
\fbox{\parbox{\textwidth}{\scriptsize
\begin{tabular}{lr}
\AxiomC{$\Delta;\cdot\vdash A\;\mathrm{type}$}
\RightLabel{}
\UnaryInfC{$\Delta;\cdot \vdash !A\;\mathrm{type}$}
\DisplayProof
& \\
& \\
\AxiomC{$\Delta;\cdot \vdash a:A$}
\RightLabel{}
\UnaryInfC{$\Delta;\cdot\vdash !a:!A$}
\DisplayProof
&\hspace{62pt}
\AxiomC{$\Delta;\Xi\vdash t:!A$}
\AxiomC{$\Delta,x:A;\Xi'\vdash b:B$}
\RightLabel{}
\BinaryInfC{$\Delta;\Xi,\Xi'\vdash \mathrm{let}\; t\;\mathrm{be}\;!x\; \mathrm{in}\; b:B$}
\DisplayProof
\\[-2pt]
&\hspace{18pt}{\scriptsize(if $\vdash\Delta;\Xi'\;\mathrm{ctxt}$ and $\Delta;\cdot \vdash B\;\mathrm{type}$)}
\\
&\\
\AxiomC{$\Delta;\Xi\vdash\mathrm{let}\; !a\;\mathrm{be}\;!x\; \mathrm{in}\; b:B$}
\RightLabel{}
\UnaryInfC{$\Delta;\Xi\vdash \mathrm{let}\; !a\;\mathrm{be}\;!x\; \mathrm{in}\; b\equiv b[{a}/x]:B$}
\DisplayProof &
\AxiomC{$\Delta;\Xi\vdash \mathrm{let}\;t\;\mathrm{be}\;!x\; \mathrm{in}\; !x:!A$}
\RightLabel{}
\UnaryInfC{$\Delta;\Xi\vdash \mathrm{let}\;t\;\mathrm{be}\;!x\; \mathrm{in}\; !x \equiv t:!A$}
\DisplayProof 
\end{tabular}
\normalsize
}}
\caption{Rules for the usual linear type formers in each context ($I$-F, -I, -E, -C, -U, $\otimes$-F, -I, -E, -C, -U, $\multimap$-F, -I, -E, -C, -U, $\top$-F, -I, -U, $\&$-F, -I, -E1, -E2, -C1, -C2, -U, $0$-F, -E, -U, $\oplus$-F, -I1, -I2, -E, -C1, -C2, -U, $!$-F, -I, -E, -C, -U).}
\end{figure}
\vspace{-10pt}
Finally, we add rules asserting all possible commuting conversions, which, from a syntactic point of view, restore the subformula property and, from a semantic point of view, say that our rules are natural transformations (between hom-functors), which simplifies the categorical semantics significantly. We represent these schematically, following \cite{barber1996dual}. That is, if $C[-]$ is a linear program context, i.e. a context built without using $!$, then (abusing notation and dealing with all the $\lbi{}{}{}$-constructors at once) the following rules hold.\vspace{-15pt}
\begin{figure}
\centering
\fbox{\parbox{\textwidth}{\scriptsize
\begin{tabular}{lr}
\AxiomC{$\Delta ;\Xi \vdash C[\lbi{a}{b}{c}]:D$}
\UnaryInfC{$\Delta ;\Xi \vdash C[\lbi{a}{b}{c}]\equiv \lbi{a}{b}{C[c]}:D$}
\DisplayProof
\quad

&\hspace{7pt}

\AxiomC{$\Delta;\Xi\vdash C[\mathrm{false}(t)]:D$}
\UnaryInfC{$\Delta;\Xi\vdash C[\mathrm{false}(t)]\equiv\mathrm{false}(t):D$}
\DisplayProof
\quad
\\
\\
if $C[-]$ does not bind any free variables in $a$ or $b$; &\hspace{7pt} if $C[-]$ does not bind any free variables in $t$;\\
\\
\\
\end{tabular}
\AxiomC{$\Delta;\Xi\vdash C[\mathrm{case}\;t\;\mathrm{of}\;\mathrm{inl}(x)\rightarrow c\; ||\; \mathrm{inr}(y)\rightarrow d] :D$}
\UnaryInfC{$\Delta;\Xi\vdash C[\mathrm{case}\;t\;\mathrm{of}\;\mathrm{inl}(x)\rightarrow c\; ||\; \mathrm{inr}(y)\rightarrow d]\equiv \mathrm{case}\;t\;\mathrm{of}\;\mathrm{inl}(x)\rightarrow C[c]\; ||\; \mathrm{inr}(y)\rightarrow C[d] :D$}
\DisplayProof
\quad\\
\\
if $C[-]$ does not bind any free variables in $t$ or $x$ or $y$.
\normalsize
}}
\caption{Commuting conversions.}
\end{figure}
\vspace{-20pt}
\begin{remark}Note that all type formers that are defined context-wise ($I$, $\otimes$, $\multimap$, $\top$, $\&$, $0$, $\oplus$, and $!$) are automatically preserved under the substitutions from Int-Ty-Subst (up to canonical isomorphism\footnote{By an isomorphism of types $\Delta;\cdot\vdash A\;\mathrm{type}$ and $\Delta;\cdot\vdash B\;\mathrm{type}$ in context $\Delta$, we mean a pair of terms $\Delta;x:A\vdash f:B$ and $\Delta;y:B\vdash g:A$ together with a pair of judgemental equalities $\Delta;x:A\vdash g[f/y]\equiv x:A$ and $\Delta;y:B\vdash f[g/x]\equiv y:B$.}), in the sense that $F(A_1,\ldots, A_n)[{a}/x]$ is isomorphic to $F(A_1[{a}/x],\ldots,A_n[{a}/x])$ for an $n$-ary type former $F$. Similarly, for $T=\Sigma$ or $\Pi$, we have that $(T_{!y:!B}C)[{a}/x]$ is isomorphic to $T_{!y:!B[{a}/x]}C[{a}/x]$ and $(\mathrm{Id}_{!B}(b,b'))[a/x]$ is isomorphic to $\mathrm{Id}_{!{B[a/x]}}(b[a/x],b'[a/x])$. This gives us Beck-Chevalley conditions in the categorical semantics.\end{remark}
\begin{remark}
The reader may note that the usual formulation of universes for DTT transfers naturally to ILDTT, giving us a notion of universes for linear types. This allows us to write rules for forming types as rules for forming terms, as usual. We do not take this approach and define the various type formers in the setting without universes.
\end{remark}
\clearpage
\subsubsection*{Some Basic Results} As the focus of this paper is the syntax-semantics correspondence, we will briefly state only a few syntactic results. For standard metatheoretic properties of (a system equivalent to) the $\multimap,\Pi,\top,\&$-fragment of our syntax, we refer the reader to \cite{cervesato1996linear}. Standard techniques and some small adaptations of the system should suffice to extend the results to all of ILDTT.

We only record that ILDTT is consistent both as a type theory (that is, not for all $\Delta;\Xi\vdash a,a':A$ do we have $\Delta;\Xi\vdash a\equiv a':A$) and as a logic (ILDTT does not prove that every type is inhabited).
\begin{theorem}[Consistency] ILDTT with all its type formers is consistent, both as a type theory and as a logic.\end{theorem}
\begin{proof}[sketch] This follows from model-theoretic considerations. Later, in Section \ref{sec:sem}, we shall see that our model theory encompasses that of DTT, for which we have models exhibiting both types of consistency.\end{proof}

To give the reader some intuition for these linear $\Pi$- and $\Sigma$-types, we suggest the following two interpretations.

\begin{theorem}[$\Pi$ and $\Sigma$ as Dependent $!(-)\multimap(-)$ and $!(-)\otimes(-)$] Suppose we have $!$-types. Let $\Delta,x:A;\cdot \vdash B\;\mathrm{type}$, where $x$ is not free in $B$. Then,
\begin{enumerate} 
\item $\Pi_{!x:{!A}}B$ is isomorphic to $!A\multimap B$, if we have $\Pi$-types and $\multimap$-types;
\item $\Sigma_{!x:{!A}}B$ is isomorphic to $!A\otimes B$, if we have $\Sigma$-types and $\otimes$-types.
\end{enumerate}
\end{theorem}
In particular, we have the following stronger version of a special case.
\begin{theorem}[$!$ as $\Sigma I$]\label{thm:!fromsigma}
Suppose we have $\Sigma$- and $I$-types. Let $\Delta;\cdot \vdash A\;\mathrm{type}$. Then, $\Sigma_{!x:{!A}}I$ satisfies the rules for $!A$. Conversely, if we have $!$- and $I$-types, then $!A$ satisfies the rules for $\Sigma_{!x:{!A}}I$. 
\end{theorem}

A second interpretation is that $\Pi$ and $\Sigma$ generalise $\&$ and $\oplus$. Indeed, the idea is that these (or their infinitary equivalents) are the forms they take when indexed over discrete types. The subtlety of this result lies in the definition of a discrete type. The same phenomenon is observed in a different context in Section \ref{sec:dismod}.

For our purposes, a discrete type is a strong sum of $\top$ (a sum with a dependent -E-rule). For simplicity, let us limit ourselves to the binary case. In this setting, the discrete type with two elements will be $2=\top\oplus \top$, where $\oplus$ has a strong/dependent -E-rule (note that this is not our $\oplus$-E). Explicitly, $2$ is a type with the following -F-, -I-, and -E-rules (and the obvious -C- and -U-rules):
\vspace{-15pt}
\begin{figure}
\fbox{\parbox{\textwidth}{\scriptsize
\begin{tabular}{lcr}
\AxiomC{}
\RightLabel{}
\UnaryInfC{$\Delta;\cdot\vdash 2\;\mathrm{type}$}
\DisplayProof
\hspace{90pt}
&
\AxiomC{}
\RightLabel{}
\UnaryInfC{$\Delta;\cdot\vdash \mathrm{tt}:2$}
\DisplayProof
\hspace{90pt}
&
\AxiomC{}
\RightLabel{}
\UnaryInfC{$\Delta;\cdot\vdash \mathrm{ff}:2$}
\DisplayProof
\end{tabular}
\\
\\
\\
\begin{tabular}{c}
\AxiomC{$\Delta,x:2;\cdot \vdash A\;\mathrm{type}$}
\AxiomC{$\Delta;\cdot\vdash t:2$}
\AxiomC{$\Delta;\Xi\vdash a_{\mathrm{tt}}:A[{\mathrm{tt}}/x]$}
\AxiomC{$\Delta;\Xi\vdash a_{\mathrm{ff}}:A[{\mathrm{ff}}/x]$}
\RightLabel{}
\QuaternaryInfC{$\Delta;\Xi\vdash \mathrm{if}\; t\;\mathrm{then}\; a_{\mathrm{tt}}\;\mathrm{else}\;a_{\mathrm{ff}}:A[{t}/x]$}
\DisplayProof
\normalsize
\end{tabular}
}}
\caption{Rules for a discrete type $2$, with -C- and -U-rules omitted for reasons of space.}
\end{figure}
\vspace{-20pt}
\begin{theorem}[$\Pi$ and $\Sigma$ as Infinitary Non-Discrete $\&$ and $\oplus$] If we have a discrete type $2$ and a type family $\Delta,x: 2;\cdot\vdash A\;\mathrm{type}$, then
\begin{enumerate}
\item $\Pi_{!x:{!2}}A$ satisfies the rules for $A[{\mathrm{tt}}/x]\& A[{\mathrm{ff}}/x]$;
\item $\Sigma_{!x:{!2}}A$ satisfies the rules for $A[{\mathrm{tt}}/x]\oplus A[{\mathrm{ff}}/x]$.
\end{enumerate}
\end{theorem}

\section{Categorical Semantics}
\vspace{-7pt}
\label{sec:sem}
We now introduce a notion of categorical model for which soundness and completeness results hold with respect to the syntax of ILDTT in the presence of $I$- and $\otimes$-types\footnote{If we are interested in the case without $I$- and $\otimes$-types, the semantics easily generalises to strict indexed symmetric multicategories with comprehension.}. This notion of model will prove to be particularly useful when thinking about various (extensional) type formers.

\begin{definition}By a \emph{strict indexed symmetric monoidal category with comprehension}, we will mean the following data.
\begin{enumerate}
\item A category $\mathcal{C}$ with a terminal object $\cdot$.
\item A strict indexed symmetric monoidal category $\mathcal{L}$ over $\mathcal{C}$, i.e. a contravariant functor $\mathcal{L}$ into the category $\mathrm{SMCat}$ of (small) symmetric monoidal categories and strong monoidal functors $\mathcal{C}^{op}\ra{\mathcal{L}}\mathrm{SMCat}.$
We will also write $-\{f\}:=\mathcal{L}(f)$ for the action of $\mathcal{L}$ on a morphism $f$ of $\mathcal{C}$.
\item A \emph{comprehension schema}: for each $\Delta\in\mathrm{ob}(\mathcal{C})$ and $A\in\mathrm{ob}(\mathcal{L}(\Delta))$, a representation of the functor\vspace{-2pt} $$x\mapsto\mathcal{L}(\mathrm{dom}(x))(I,A\{x\}):(\mathcal{C}/\Delta)^{op}\ra{}\mathrm{Set}.$$ \vspace{-2pt} We will write its representing object\footnote{Really, $\Delta.MA\ra{\mathbf{p}_{\Delta,MA}}\Delta$ would be a better notation, where we think of $L\dashv M$ as an adjunction inducing $!$, but it would be very verbose.} $\Delta.{A}\ra{\mathbf{p}_{\Delta,{A}}}\Delta\in\mathrm{ob}(\mathcal{C}/\Delta)$ and universal element $\mathbf{v}_{\Delta,{A}}\in\mathcal{L}(\Delta.{A})(I,A\{\mathbf{p}_{\Delta,{A}}\})$. We will write $a  \mapsto  \langle f,a\rangle$ for the isomorphism $\mathcal{L}(\Delta')(I,A\{f\}) \cong  \mathcal{C}/\Delta(f,\mathbf{p}_{\Delta,{A}})$, if $\Delta'\ra{f}\Delta$.
\end{enumerate}
\end{definition}
\begin{remark}Note that this notion of model reduces to a standard notion of model for DTT when the monoidal structures on the fibre categories are Cartesian: a reformulation of split comprehension categories with $1$- and $\times$-types. To get a precise fit with the syntax, one usually imposes the extra condition called ``fullness'' on these \cite{jacobs1993comprehension}. The fact that we leave out this last condition precisely allows for non-trivial $!$-types (i.e. ones such that $!A\ncong A$) in our models of ILDTT. Every model of DTT is, in particular, a (degenerate) model of ILDTT. We will see that the type formers of ILDTT also generalise those of DTT.\end{remark}
\begin{theorem}[Soundness] We can soundly interpret ILDTT with $I$- and $\otimes$-types in a strict indexed symmetric monoidal category $(\mathcal{C},\mathcal{L})$ with comprehension.
\end{theorem}
\begin{proof}[sketch]
The idea is that a context $\Delta;\Xi$ will be (inductively) interpreted by a pair of objects $\sem{\Delta}\in\mathrm{ob}(\mathcal{C})$, $\sem{\Xi}\in\mathrm{ob}(\mathcal{L}(\sem{\Delta}))$, a type $A$ in context $\Delta;\cdot$ by an object $\sem{A}$ of $\mathcal{L}(\sem{\Delta})$, and a term $a:A$ in context $\Delta;\Xi$ by a morphism $\sem{\Xi}\ra{\sem{a}}\sem{A}\in\mathcal{L}(\sem{\Delta})$. Generally, the interpretation of the propositional linear type theory in intuitionistic context $\Delta;\cdot$ takes place in $\mathcal{L}(\sem{\Delta})$, as would be expected.

The crux is that Int-C-Ext ($\sem{\Delta,x:A}:=\mathrm{dom}(\mathbf{p}_{\sem{\Delta},\sem{A}})$), Int-Var ($\sem{\Delta,x:A;\cdot\vdash x:A}:=\mathbf{v}_{\sem{\Delta},\sem{A}}$), and Int-Subst (by $\mathcal{L}(\langle \mathrm{id}_{\sem{\Delta}},\sem{a}\rangle)$) are interpreted through the comprehension, as is Int-Weak (through $\mathcal{L}$ of the obvious morphism in $\mathcal{C}$).

The remaining soundness checks are straightforward.
\end{proof}
\begin{theorem}[Completeness] In fact, this interpretation is complete.\end{theorem}
\begin{proof}[sketch] This follows from the construction of a syntactic category.\end{proof}
In fact, we would like to say that the syntax is even an internal language for such categories. This is almost true, and can be made entirely true either by imposing on our notion of model a restriction that excludes any non-trivial morphisms into objects that are not of the form $\Delta.{A}$, or by extending the syntax to talk about context morphisms explicitly \cite{pitts2001categorical}. Following the DTT tradition, we have opted against the latter.\\
\\
We next characterise the categorical description of the various type formers. First, we note the following.
\begin{theorem}[Comprehension Functor]\label{thm:comprfunc} A comprehension schema $(\mathbf{p},\mathbf{v})$ on a strict indexed symmetric monoidal category $(\mathcal{C},\mathcal{L})$ defines a morphism $\mathcal{L}\ra{M}\mathcal{I}$ of indexed categories, where $\mathcal{I}$ is the full sub-indexed category of $\mathcal{C}/-$ (after making a choice of pullbacks) on the objects of the form $\mathbf{p}_{\Delta,{A}}$ and where\\
\[
\begin{tikzcd}[column sep=large]
M_\Delta(A\ra{a}B):=\mathbf{p}_{\Delta,{A}} \arrow[r,"{\langle\mathbf{p}_{\Delta,{A}},a\{\mathbf{p}_{\Delta,{A}}\}\circ \mathbf{v}_{\Delta,{A}}\rangle}"] & \mathbf{p}_{\Delta,{B}}.
\end{tikzcd}
\]
\end{theorem}

Note that $\mathcal{I}$ is a display map category and hence a model of DTT \cite{jacobs1993comprehension}. We will think of it as the intuitionistic content of $\mathcal{L}$. We will see that the comprehension functor gives us a unique candidate for $!$-types: $!:=LM$, where $L\dashv M$ is a monoidal adjunction. We conclude that, in ILDTT, the $!$-modality is uniquely determined by the indexing. This is worth noting, because, in propositional linear type theory, we might have many different candidates for $!$-types.  

\begin{theorem}[Semantic Type Formers]\label{thm:semtype} For the other type formers, a model $(\mathcal{C},\mathcal{L},\mathbf{p},\mathbf{v})$ of ILDTT with $I$- and $\otimes$-types has the following properties.
\begin{enumerate}
\item It supports $\Sigma$-types iff all the pullback functors $\mathcal{L}(\mathbf{p}_{\Delta,{A}})$ have left adjoints $\Sigma_{!{A}}$ that satisfy the Beck-Chevalley condition in the sense that the canonical map $\Sigma_{!{A\{f\}}}\circ \mathcal{L}(\mathbf{q}_{f,{A}})\ra{}\mathcal{L}(f)\circ\Sigma_{!{A}}$ is an iso, where $\Delta'\ra{f}\Delta$ and $\mathbf{q}_{f,{A}}:=\langle f\circ\mathbf{p}_{\Delta',{A\{f\}}},\mathbf{v}_{\Delta',{A\{f\}}} \rangle$, and that satisfy Frobenius reciprocity in the sense that the canonical morphism $\Sigma_{!{A}}(\Xi'\{\mathbf{p}_{\Delta,{A}}\}\otimes B)\ra{} \Xi'\otimes \Sigma_{!{A}}B$ is an isomorphism, for all $\Xi'\in\mathcal{L}(\Delta)$, $B\in\mathcal{L}(\Delta.{A})$.
\item It supports $\Pi$-types iff all the pullback functors $\mathcal{L}(\mathbf{p}_{\Delta,{A}})$ have right adjoints $\Pi_{!{A}}$ that satisfy the dual Beck-Chevalley condition for the pullback squares described above: the canonical $\mathcal{L}(f)\circ\Pi_{!{A}}\ra{}\Pi_{!{A\{f\}}}\circ\mathcal{L}(\mathbf{q}_{f,{A}})$ is an iso.
\item It supports $\multimap$-types iff $\mathcal{L}$ factors over the category $\mathrm{SMCCat}$ of symmetric monoidal closed categories and their homomorphisms.
\item It supports $\top$- and $\&$-types iff $\mathcal{L}$ factors over the category of Cartesian categories with symmetric monoidal structure and their homomorphisms.
\item It supports $0$- and $\oplus$-types iff $\mathcal{L}$ factors over the category $\mathrm{dSMcCCat}$ of co-Cartesian categories with a distributive symmetric monoidal structure and their homomorphisms.
\item If it supports $\multimap$-types, it supports $!$-types iff all the comprehension functors $\mathcal{L}(\Delta)\ra{M_\Delta}\mathcal{I}(\Delta)$ have a strong monoidal left adjoint $\mathcal{I}(\Delta)\ra{L_\Delta}\mathcal{L}(\Delta)$ and $L_-$ is a morphism of indexed categories: for all $\Delta'\ra{f}\Delta\in \mathcal{C}$, $L_{\Delta'}\mathcal{I}(f)=\mathcal{L}(f)L_\Delta$. Then $!_\Delta:=L_\Delta\circ M_\Delta$ interprets the comodality $!$ in context $\Delta$.
\item If it supports $\multimap$-types, it supports $\mathrm{Id}$-types iff for all $A\in\mathrm{ob}(\mathcal{L}(\Delta))$, we have left adjoints $\mathrm{Id}_{!A}\dashv -\{\mathrm{diag}_{\Delta,{A}}\}$ that satisfy Frobenius reciprocity for $\mathrm{diag}_{\Delta,A}$ and a Beck-Chevalley condition: $\mathrm{Id}_{!{A\{f\}}}\circ \mathcal{L}(\mathbf{q}_{f,{A}})\ra{}\mathcal{L}(\mathbf{q}_{\mathbf{q}_{f,{A}},{A\{\mathbf{p}_{\Delta,A}\}}})\circ \mathrm{Id}_{!A}$ is an iso. Here, $\mathrm{Id}_{!A}(I)$ interprets $\mathrm{Id}_{!A}(x,x')$, where
\[
\begin{tikzcd}[column sep=large]
\Delta.A \arrow[r,"{\mathrm{diag}_{\Delta,A}:=\langle \mathrm{id}_{\Delta.A},\mathbf{v}_{\Delta,A}\rangle}"] & \Delta.A.A\{\mathbf{p}_{\Delta,A}\}.
\end{tikzcd}
\]
\end{enumerate}
\end{theorem}

The semantics of $!$ suggests an alternative definition of the notion of comprehension: if we have $\Sigma$-types in a strong sense, it is a derived notion.

\begin{theorem}[Lawvere Comprehension]\label{altcompr} If a strict indexed monoidal category $(\mathcal{C},\mathcal{L})$ has left adjoints $\Sigma_f$ to $\mathcal{L}(f)$ for arbitrary $\Delta'\ra{f}\Delta\in\mathcal{C}$, then we can define $\mathcal{C}/\Delta\ra{L_\Delta}\mathcal{L}(\Delta)$ by
$L_\Delta(g):=\Sigma_g I$ for $g\in\mathcal{C}/\Delta$.
In that case, $(\mathcal{C},\mathcal{L})$ has a comprehension schema iff $L_\Delta$ has a right adjoint $M_\Delta$ (for which $M_{\Delta'}\circ \mathcal{L}(f)=f^*\circ M_\Delta$ for all $\Delta'\ra{f}\Delta\in\mathcal{C}$, where $f^*$ denotes pullback along $f$). That is, our notion of comprehension generalises that of Lawvere \cite{lawvere1970equality}. Finally, if the $\Sigma_f$ satisfy Frobenius reciprocity and Beck-Chevalley, then $(\mathcal{C},\mathcal{L})$ supports comprehension iff it supports $!$-types.
\end{theorem}
\begin{proof}[sketch] This follows by writing out both the representability condition defining a comprehension and the adjointness condition for $\Sigma_{f}$.\end{proof}
We observe the following about the usual intuitionistic type formers in $\mathcal{I}$.

\begin{theorem}[Type Formers in $\mathcal{I}$]\label{thm:inttyp}$\mathcal{I}$ supports $\Sigma$-types iff $\mathrm{ob}(\mathcal{I})\subset \mathrm{mor}(\mathcal{C})$ is closed under binary compositions. $\mathcal{I}$ supports $\mathrm{Id}$-types iff $\mathrm{ob}(\mathcal{I})$ is closed under post-composition with $\mathrm{diag}_{\Delta,A}$. If $\mathcal{L}$ supports $!$- and $\Pi$-types, then $\mathcal{I}$ supports $\Pi$-types. Moreover, type formers in $\mathcal{I}$ relate to those in $\mathcal{L}$ as follows, leaving out the subscripts of the indexed functors $L\dashv M$:\vspace{-3pt}
$$\Sigma_{!A}! B\cong L(\Sigma_{M A}MB) \qquad \mathrm{Id}_{!A}(!B)\cong L\mathrm{Id}_{MA}(MB)\qquad M\Pi_{!B}C\cong \Pi_{MB}MC.
$$
\end{theorem}
\begin{remark}[Dependent Seely Isomorphisms?] It is easily seen that $M_\Delta(\top)\cong\mathrm{id}_\Delta$ and $M_\Delta(A\& B)\cong M_\Delta(A)\times M_\Delta(B)$, and hence $!_\Delta\top\cong I$ and $!_\Delta(A\& B)\cong !_\Delta A\otimes !_\Delta B$.

Now, Theorem \ref{thm:inttyp} suggests similar Seely isomorphisms for $\Sigma$- and $\mathrm{Id}$-types. Indeed, $\mathcal{I}$ supports $\Sigma$-types, respectively $\mathrm{Id}$-types, iff we have ``additive'' $\Sigma$-types, respectively $\mathrm{Id}$-types, that is, $\Sigma_A^\& B,\mathrm{Id}_A^\&(B)\in\mathrm{ob}(\mathcal{L})$ such that
$$M\Sigma_A^\& B\cong \Sigma_{MA}MB\txt{and hence} !\Sigma_A^\& B\cong \Sigma_{!A}^\otimes !B\txt{respectively}$$
$$M\mathrm{Id}_A^\&(B)\cong \mathrm{Id}_{MA}(MB)\txt{and hence} !\mathrm{Id}_A^\&(B)\cong \mathrm{Id}_{!A}^\otimes(!B),$$
where we write $\Sigma^\otimes$ and $\mathrm{Id}^\otimes$ for the usual multiplicative $\Sigma$- and $\mathrm{Id}$-types\footnote{We call the usual $\mathrm{Id}$-types ``multiplicative'' connectives, for instance because $\mathrm{Id}^\otimes_{!A}(B)\cong \mathrm{Id}^\otimes_{!A}(I)\otimes B$. Similarly, if we have a suitable $\mathrm{Id}^\&_{A}(\top)$, we can define $\mathrm{Id}^\&_{A}(B):= \mathrm{Id}^\&_A(\top)\& B$.}.

This situation arises, and such additive $\Sigma$- and $\mathrm{Id}$-types have to be considered when $L_\cdot \dashv M_\cdot: \mathcal{L}(\cdot)\ra{}\mathcal{C}$ is the co-Kleisli adjunction of $!$. See \cite{vakar2014syntax} for more discussion.\end{remark}
\section{Some Discrete Models: Monoidal Families}
\vspace{-7pt}
\label{sec:dismod}
We discuss a simple class of models in terms of families with values in a symmetric monoidal category. On a logical level, the construction amounts to starting with a model $\mathcal{V}$ of a linear propositional logic and taking the cofree linear predicate logic on $\mathrm{Set}$ with values in this propositional logic. This important example illustrates how $\Sigma$- and $\Pi$-types can represent infinitary additive disjunctions and conjunctions. The model is discrete in nature, however, and, in that respect, is not representative of ILDTT.

Suppose $\mathcal{V}$ is a symmetric monoidal category. We can then consider a strict $\mathrm{Set}$-indexed category, defined through the following enriched Yoneda embedding $\mathrm{Fam}(\mathcal{V}):={\mathcal{V}}^{-}:=\mathrm{SMCat}(-,\mathcal{V})$:
\[
\begin{tikzcd}[column sep=normal,ampersand replacement=\&]
\mathrm{Set}^{op} \arrow[r,"{\mathrm{Fam}(\mathcal{V})}"] \& \mathrm{SMCat} \& \& S\ra{f}S' \arrow[r,mapsto] \& \mathcal{V}^S\stackrel{-\circ f}{\longleftarrow} \mathcal{V}^{S'}.
\end{tikzcd}
\]
Note that this definition naturally extends to a functorial embedding $\mathrm{Fam}$.
\begin{theorem}[Families Model ILDTT] The construction $\mathrm{Fam}$ adds type dependency over $\mathrm{Set}$ cofreely, in the sense that it is right adjoint to the forgetful functor $\mathrm{ev}_1$ that evaluates a model of linear dependent type theory at the empty context to obtain a model of linear propositional type theory (where $\mathrm{SMCat}_{\mathrm{compr}}^{\mathrm{Set}^{op}}$ is the full subcategory of $\mathrm{SMCat}^{\mathrm{Set}^{op}}$ on the objects with comprehension):
\[
\begin{tikzcd}[column sep=huge,ampersand replacement=\&]
\mathrm{SMCat} \arrow[r,hook,shift right=1.5ex,"{\mathrm{Fam}}"'] \& \mathrm{SMCat}_{\mathrm{compr}}^{\mathrm{Set}^{op}}. \arrow[l,shift right=1.5ex,"{\mathrm{ev}_1}"']
\arrow[from=1-1,to=1-2,phantom,"\bot" description]
\end{tikzcd}
\]
\end{theorem}
\begin{proof}[sketch] The comprehension on $\mathrm{Fam}(\mathcal{V})$ is given by the obvious bijection
$$\mathrm{Fam}(\mathcal{V})(S)(I,B\{f\})\cong \prod_{s\in S}\mathcal{V}(I,B(f(s)))  \cong \mathrm{Set}/S'(f,\mathbf{p}_{S',{B}}),$$
where $\mathbf{p}_{S',{B}}:=\coprod_{s'\in S'}\mathcal{V}(I,B(s'))\ra{\mathrm{fst}}S'$. The rest of the proof is a straightforward verification, where the adjunction relies on $\mathrm{Set}$ being well-pointed.\end{proof}
We express the existence of type formers in $\mathrm{Fam}(\mathcal{V})$ as conditions on $\mathcal{V}$. A characterisation of additive $\Sigma$- and $\mathrm{Id}$-types can be found in \cite{vakar2014syntax}.

\begin{theorem}[Type Formers for Families]$\mathcal{V}$ has small coproducts that distribute over $\otimes$ iff $\mathrm{Fam}(\mathcal{V})$ supports $\Sigma$-types. In that case, $\mathrm{Fam}(\mathcal{V})$ also supports $0$- and $\oplus$-types (which correspond precisely to finite distributive coproducts).

$\mathcal{V}$ has small products iff $\mathrm{Fam}(\mathcal{V})$ supports $\Pi$-types. In that case, $\mathrm{Fam}(\mathcal{V})$ also supports $\top$- and $\&$-types (which correspond precisely to finite products).

$\mathrm{Fam}(\mathcal{V})$ supports $\multimap$-types iff $\mathcal{V}$ is monoidal closed. 

$\mathrm{Fam}(\mathcal{V})$ supports $!$-types iff $\mathcal{V}$ has small coproducts of $I$ that are preserved by $\otimes$ in the sense that the canonical morphism $\coprod_S(\Xi'\otimes I)\ra{}\Xi'\otimes \coprod_S I$ is an isomorphism for any $\Xi'\in\mathrm{ob}\;\mathcal{V}$ and $S\in\mathrm{ob}\;\mathrm{Set}$. In particular, if $\mathrm{Fam}(\mathcal{V})$ supports $\Sigma$-types, then it also supports $!$-types.

$\mathrm{Fam}(\mathcal{V})$ supports $\mathrm{Id}$-types if $\mathcal{V}$ has an initial object. If $\mathcal{V}$ has a terminal object, the only-if direction also holds.
\end{theorem}
\begin{proof}[sketch] We supply some definitions and leave the rest to the reader.

$\top$-, $\&$-, $0$-, and $\oplus$-types are clear, since (co)limits are pointwise in a functor category. $\multimap$-types also follow immediately from the previous section. We define $\Sigma_f(A)(s'):=\coprod_{s\in f^{-1}(s')}A(s)$ and $\Pi_f(A)(s')=\prod_{s\in f^{-1}(s')}A(s)$. Then $\Sigma_f\dashv -\{f\} \dashv \Pi_f$. We define $\mathrm{Id}_{!A}(B)(s,a,a'):=\left\{\begin{array}{l}B(s,a)\textnormal{ if $a=a'$}\\ 0\textnormal{ else} \end{array}\right.$. Then, $\mathrm{Id}_{!A}\dashv -\{\mathrm{diag}_{S,A}\}$. The Beck-Chevalley conditions follow from the fact that substitution is interpreted as precomposition. Finally, this leads to the definition $!A(s):=\coprod_{\mathcal{V}(I,A(s))}I$, which only depends on $A(s)$.
\end{proof}

\begin{remark}Note that an obvious way to guarantee distributivity of coproducts over $\otimes$ is by demanding that $\mathcal{V}$ be monoidal closed.
\end{remark}\vspace{-2pt}
Two simple concrete choices of $\mathcal{V}$ accommodate all type formers and illustrate real \emph{linear} type dependency: the category $\mathcal{V}=\mathrm{Vect}_F$ of vector spaces over a field $F$, with the tensor product, and the category $\mathcal{V}=\mathrm{Set}_*$ of pointed sets, with the smash product. All type formers have their obvious interpretations, but let us consider $!$: a novelty of ILDTT is that it is uniquely determined by the indexing, while in propositional linear logic we might have several choices. In the first example, $!$ is given by: $(!B)(s')=\coprod_{ \mathrm{Vect}_F(F,B(s'))}F\cong   \bigoplus_{ B(s')}F$, i.e. the vector space freely spanned by all vectors. In the second example, $(!B)(s')=\coprod_{\mathrm{Set}_*(2_*,B(s'))}2_*=\bigvee_{B(s')}2_*=B(s')+\{*\}$, i.e. $!$ freely adds a new basepoint. These models show the following.
\begin{theorem}[DTT, DILL$\subsetneq$ ILDTT] ILDTT is a proper generalisation of DTT and DILL: we have inclusions of the classes of models DTT, DILL$\subsetneq$ILDTT.\end{theorem}\vspace{-4pt}
Although this class of models is important, it is clear that it represents only a limited part of the generality of ILDTT. Hence, we need non-Cartesian models that are less discrete in nature if we hope to observe interesting new phenomena arising from the connectives of linear dependent type theory. Some suggestions and work in progress will be discussed in the next section.
\vspace{-4pt}
\section{Conclusions and Future Work}\vspace{-7pt}
We hope to have convinced the reader that linear dependent types fit naturally into the landscape of existing type theories and that they admit a well-behaved semantic theory.

We have presented a system, ILDTT, that, on a syntactic level, is a natural blend between (intuitionistic) dependent type theory (DTT) and dual intuitionistic linear logic (DILL). On a semantic level, if one starts with the right notion of model for dependent types, the linear generalisation is obtained by the usual passage from Cartesian to symmetric monoidal structures. The resulting notion of a model forms a natural blend between comprehension categories, modelling DTT, and linear-non-linear models of DILL.

It is notable that all the syntactically natural rules for type formers are equivalent to the semantic counterparts one would expect from the traditions of categorical logic for dependent and linear types. In particular, from the point of view of logic, it is interesting to see that the categorical semantics seems to have a preference for multiplicative quantifiers.

Finally, we have shown that, as in the intuitionistic case, we can represent infinitary (additive) disjunctions and conjunctions in linear type theory, through cofree $\Sigma$- and $\Pi$-types, indexed over $\mathrm{Set}$. In particular, this construction exhibits a family of non-trivial, truly linear models of dependent types. Moreover, it shows that ILDTT properly extends both DILL and DTT.

Despite what might be expected from this paper, much of this work has been strongly motivated by specific semantic models. In joint work with Samson Abramsky, a model of linear dependent types with comprehension has been constructed in a category of coherence spaces. Apart from the usual type constructors from linear logic, it also supports $\Sigma$-, $\Pi$-, and $\mathrm{Id}$-types. A detailed account of this model will be made available soon.

Besides providing a first non-trivial, semantically motivated model of such a type system that goes properly beyond DILL and DTT, this work served as a stepping stone for a model in a category of games, developed together with Radha Jagadeesan and Samson Abramsky. In particular, this provides a game semantics for dependent type theory. 

An indexed category of spectra over topological spaces has been studied as a setting for stable homotopy theory \cite{may2006parametrized,ponto2012duality}. It has been shown to admit $I$-, $\otimes$-, $\multimap$-, and $\Sigma$-types. The natural candidate for a comprehension adjunction here is that between the infinite suspension spectrum and the infinite loop space: $L\dashv M\;\;=\;\;\Sigma^\infty\dashv \Omega^\infty$. A detailed examination of the situation and an explanation of the relation with the Goodwillie calculus is desirable. This might fit in with our ultimate objective of a linear analysis of homotopy type theory.

Another fascinating possibility is that of models related to quantum mechanics. Non-dependent linear type theory has found interesting interpretations in quantum computation \cite{AbrDun:CQLv2:2004}. The question arises whether the extension to dependent linear types has a natural counterpart in physics and could, for example, provide stronger type systems for quantum computing. Also suggestive is Schreiber's work \cite{schreiber2014quantization}, which sketches how linear dependent types can serve as a language to talk about quantum field theory and quantisation in particular.

Finally, there are still many theoretical questions about the type theory. Can we find interesting models with type dependency on the co-Kleisli category of $!$, and can we make sense of additive $\Sigma$- and $\mathrm{Id}$-types, for example from the point of view of syntax? Or should we perhaps doubt the canonicity of the Girard translation and accept that dependent types are more naturally modelled in co-Eilenberg-Moore categories? Is there an analogue of strong/dependent E-rules for ILDTT and how do we model interesting intensional $\mathrm{Id}$-types? Does the Curry-Howard correspondence extend in its full glory: do we have a propositions-as-types interpretation of linear predicate logic in ILDTT? These questions call for a combination of research into the formal system and the study of specific models. We hope that the general framework sketched here will play its part in connecting all the different sides of the story: from syntax to semantics; from computer science and logic to geometry and physics.
\subsection*{Acknowledgements}
\vspace{-5pt}
My thanks go to Samson Abramsky and Radha Jagadeesan for the stimulating discussions and to Urs Schreiber for sparking my curiosity about this topic. I am indebted to the anonymous reviewers, whose comments have been very helpful. This research was supported by the EPSRC and the Clarendon Fund.
\vspace{-12pt}
\iffalse

\fi

\scriptsize
\bibliographystyle{splncs}
\bibliography{tau}

\begin{thebibliography}{4}
\bibitem{AbrDun:CQLv2:2004} S. Abramsky and R. Duncan: \textit{A categorical quantum logic.} Mathematical Structures in Computer Science / Volume 16 / Issue 03 / June 2006.
\bibitem{barber1996dual} A. Barber: \textit{Dual Intuitionistic Linear Logic.} 1996.
\bibitem{bierman1994intuitionistic} G.M. Bierman: \textit{On Intuitionistic Linear Logic.} 1994.
\bibitem{blute1994fock} R. Blute and P. Panangaden: \textit{Proof Nets as Formal Feynman Diagrams.} In: New Structures for Physics, 2010.
\bibitem{cervesato1996linear} I. Cervesato and F. Pfenning: \textit{A Linear Logical Framework.} 1997.
\bibitem{church1940formulation} A. Church: \textit{A Formulation of the Simple Theory of Types.} JSL 5, 1940.
\bibitem{dal2011linear} U. Dal Lago and M. Gaboardi: \textit{Linear Dependent Types and Relative Completeness.} In: Logical Methods in Computer Science Vol. 8(4:11), 2012.
\bibitem{petit2012linear} U. Dal Lago and B. Petit: \textit{Linear dependent types in a call-by-value scenario.} In: Science of Computer Programming 84, 2014.
\bibitem{gaboardi2013linear}  M. Gaboardi et al.: \textit{Linear Dependent Types for Differential Privacy.} In: POPL '13, 2013.
\bibitem{girard1987linear} J.-Y. Girard: \textit{Linear Logic.} In: Theoretical Computer Science 50, 1987.
\bibitem{hofmann1997syntax} M. Hofmann: \textit{Syntax and Semantics of Dependent Type Theory.} In: Semantics and Logics of Computation 1997.
\bibitem{howard1995formulae} W. Howard: \textit{The formulae-as-types notion of construction.} In: J. Roger Seldin, Jonathan P.; Hindley, (ed.s), To H.B. Curry: Essays on Combinatory Logic, Lambda Calculus and Formalism. 1980. original paper manuscript from 1969.
\bibitem{jacobs1993comprehension} B. Jacobs: \textit{Comprehension categories and the semantics of type dependency.} 1993.
\bibitem{lambek1972deductive} J. Lambek: \textit{Deductive systems and categories, III.} In: Toposes, Algebraic Geometry, and Logic
(F. W. Lawvere, ed.), Springer-Verlag, 1972.
\bibitem{lawvere1970equality} F.W. Lawvere: \textit{Equality in hyperdoctrines and comprehension schema as an adjoint functor.} 1970.
\bibitem{martin1998intuitionistic} P. Martin-L\"of: \textit{An intuitionistic theory of types.} In: Twenty-five years of constructive type theory. Venice, 1995. 
\bibitem{may2006parametrized} J.P. May and J. Sigurdsson: \textit{Parametrized Homotopy Theory.} 2006.
\bibitem{mellies2009categorical} P.A. Mellies: \textit{Categorical Semantics of Linear Logic.} In: Interactive models of computation and program behaviour. Pierre-Louis Curien, Hugo Herbelin, Jean-Louis Krivine, Paul-André Melliès. Panoramas et Synthèses 27, Société Mathématique de France, 2009. 
\bibitem{o1999logic} P.W. O'Hearn and D.J. Pym: \textit{The Logic of Bunched Implications.} The Bulletin of Symbolic Logic, Vol. 5, No. 2. 1999.
\bibitem{palmgren1990domain} E. Palmgren and V. Stoltenberg-Hansen: \textit{Domain Interpretations of Martin-L\"of's Partial Type Theory.} 1990.
\bibitem{pitts2001categorical} A. Pitts: \textit{Categorical Logic.} In: Handbook for Logic in Computer Science Vol. VI. 1995.
\bibitem{ponto2012duality} K. Ponto and M. Shulman: \textit{Duality and Traces in Indexed Monoidal Categories.} 2012.
\bibitem{power1989general} A.J. Power: \textit{A general coherence result.} JPAA Vol. 57 Iss. 2, 1989.
\bibitem{schreiber2014quantization} U. Schreiber: \textit{Quantization via Linear Homotopy Types.} 2014.
\bibitem{shulman2013enriched} M. Shulman: \textit{Enriched Indexed Categories.} 2008.
\end{thebibliography}

\end{document}